\documentclass[a4paper,12pt]{article}
\usepackage{test,booktabs,tabularx}

\usepackage{hyperref}
\hypersetup{colorlinks,citecolor=blue,urlcolor=magenta}
\usepackage{doi}

\usepackage[style=authoryear-comp,
sorting=nyt,
dashed=false, 
maxcitenames=2, 
maxbibnames=99, 
uniquelist=false,
uniquename=false,
giveninits=true, 
natbib, 
date=year 
]{biblatex}

\AtBeginRefsection{\GenRefcontextData{sorting=ynt}}
\AtEveryCite{\localrefcontext[sorting=ynt]}

\DeclareFieldFormat{pages}{#1} 
\renewbibmacro{in:}{\ifentrytype{article}{}{\printtext{\bibstring{in}\intitlepunct}}} 

\usepackage[inline,shortlabels]{enumitem}
\setlist[enumerate,1]{label=(\roman*)}

\usepackage[margin = 1.25in]{geometry}
\usepackage{setspace}
\title{Rational Bubbles: A Clarification\thanks{We would like to thank Jos\'e Scheinkman and Joseph Stiglitz for continuous discussions that were very helpful for writing this draft. We also thank Alberto Martin and Jaume Ventura for their comments on the results and the exposition of this draft, and Hiroyuki Takahashi for outstanding research assistance.}}
\author{Tomohiro Hirano\thanks{Department of Economics, Royal Holloway, University of London. Email: \href{mailto:tomohiro.hirano@rhul.ac.uk}{tomohih@gmail.com}.} \and Alexis Akira Toda\thanks{Department of Economics, Emory University. Email: \href{mailto:alexis.akira.toda@emory.edu}{alexis.akira.toda@emory.edu}.}}

\begin{document}
\maketitle

\begin{abstract}
``Rational bubble'', as introduced by the famous paper on money by \citet{Samuelson1958}, means speculation backed by nothing. The large subsequent rational bubble literature has identified attaching bubbles to dividend-paying assets in a natural way as an important but challenging question. \citet{MiaoWang2018} claim to ``provide a theory of rational stock price bubbles''. Contrary to their claim, the present comment proves the nonexistence of rational bubbles in the model of \citet{MiaoWang2018}. We also clarify the precise mathematical definition and the economic meaning of ``rational bubble'' in an accessible way to the general audience. 

		
		
		
\end{abstract}

\section{Introduction}

There is significant heterogeneity in the usage of the word ``bubble'' in the economic literature. The most common and widely accepted definition, which we will adhere to and elaborate upon below, is a situation in which ``asset prices do not reflect fundamentals'' \citep{Stiglitz1990}, or in other words, the asset price ($P$) exceeds its fundamental value ($V$) defined by the present value of dividends ($D$). The condition $P>V$ could arise for various reasons. For instance, if agents hold heterogeneous beliefs or asymmetric information, the fundamental value $V$ need not be common across agents (whereas the market price $P$ is common across agents), and hence it could be $P>V$ for some or even all agents \citep{ScheinkmanXiong2003,FostelGeanakoplos2012Tranching,Barlevy2014,AllenBarlevyGale2022}. In infinite-horizon general equilibrium models, $P>V$ could hold even with rational expectations (common beliefs and complete information) because in equilibrium, agents believe that they can resell in the future the overpriced asset to other agents, \ie, there is a speculative motive. We refer to this class of models as ``rational bubbles'', which were introduced and studied by the seminal papers of \citet{Samuelson1958}, \citet{Bewley1980}, \citet{Tirole1985}, \citet{ScheinkmanWeiss1986}, \citet{Kocherlakota1992}, and \citet{SantosWoodford1997}, as well as the large subsequent literature reviewed in \citet{MartinVentura2018} and \citet{HiranoToda2024JME}. Put simply, as we explain in \S\ref{subsec:definition}, rational bubbles mean speculation: investors rationally purchase the asset at a high price for the purpose of reselling in the future, rather than (or in addition to) receiving dividends.

Some authors (including the popular press) use the word ``bubble'' loosely to describe booms and busts in asset prices. For instance, several empirical papers including \citet{JordaSchularickTaylor2015} define a bubbly episode as a deviation from the trend that is larger than some pre-specified threshold. In the time series econometrics literature \citep{PhillipsShiYu2015,PhillipsShi2020}, a bubble is defined as an explosive dynamics of the price-dividend ratio. In monetary theory \citep{LagosRocheteauWright2017}, bubble is synonymous to liquidity premium. As such, it is natural that there are diverse views and definitions regarding asset price bubbles, and it is important to understand asset prices from a broader perspective.

The primary purpose of this comment is to clarify the meaning of ``rational bubbles'' and important questions in the rational bubble literature. We need to correctly understand the following two points. The first is the meaning of rational bubbles, which is  speculation backed by nothing. In his famous paper, \citet{Samuelson1958} proved that money, which is an asset without any dividends and any backing, can still circulate in equilibrium, meaning that money is bought and sold in equilibrium solely on the basis of speculative factors despite the fact that there is nothing that backs it. This paper is revolutionary because it created the rational bubble literature and monetary economics. In rational bubble models such as \citet{Samuelson1958} and \citet{Tirole1985}, the aspect of speculation is captured by the transversality condition for asset pricing. In any asset pricing model in which the asset can be resold forever, the transversality condition for asset pricing emerges.\footnote{The transversality condition for asset pricing is different from the transversality condition for optimality. See Appendix \ref{subsec:TVC} for detail.} Importantly, the violation of the transversality condition means the existence of speculation in asset prices. Although there often appears to be a misunderstanding that the condition is a mere technicality, one needs to correctly understand that it is \emph{not}; rather, it has a clear economic meaning. This is the critically important insight \citet{Samuelson1958} discovered in asset pricing models. Hence, ``macroeconomics of rational bubbles'' can be rephrased as ``macroeconomics with speculation''.

The second point to be understood correctly is that although the definition and economic meaning of rational bubbles is the same for money and for real assets yielding positive dividends, there is a \emph{discontinuity} in proving the existence of speculation contained in equilibrium asset prices between the two cases. As we explain in \S\ref{subsec:difficulty}, it is well known from the results of \citet{Kocherlakota1992} and \citet{SantosWoodford1997} that there is a fundamental difficulty in proving the existence of speculation attached to an asset with positive dividends. Indeed, in workhorse macro-finance models, there is no aspect of speculation in asset prices in equilibrium. Therefore, the macro-finance literature including rational bubbles has regarded attaching bubbles to dividend-paying assets in a natural way as a long-standing and important open question.\footnote{On this point, we thank Jos\'e Scheinkman and Nobuhiro Kiyotaki for pointing out the limitations of models with zero dividends and teaching us how difficult and how valuable it is to prove the existence of rational bubbles attached to real assets with positive dividends.} 

The paper by \citet{MiaoWang2018} published in \emph{American Economic Review} claims to ``provide a theory of rational stock price bubbles'' (quoted from the abstract). As it is, their paper gives the impression that they have proved the existence of speculation, \ie, rational bubbles, attached to real assets and resolved the fundamental difficulty discussed above. Indeed, at various places, a number of papers on rational bubbles are cited and compared, including seminal papers by \citet{Samuelson1958} and \citet{Tirole1985}. For instance, by citing basic theory papers on rational bubbles, \citet[p.~2595]{MiaoWang2018} state
\begin{quote}
    Some studies (e.g., \citealp{ScheinkmanWeiss1986}; \citealp{Kocherlakota1992,Kocherlakota2008}; \citealp{SantosWoodford1997}; \citealp{HellwigLorenzoni2009}) have found that infinite-horizon models of endowment economies with borrowing constraints can generate rational bubbles. Unlike this literature, our paper analyzes a production economy with stock price bubbles attached to productive firms.  
\end{quote}
Concerning the bubble existence condition, \citet[p.~2607]{MiaoWang2018} state 
\begin{quote}
    We can also show that the bubbleless equilibrium is dynamically efficient in our model. [\ldots] Thus, the condition that the economy must be dynamically inefficient in \citet{Tirole1985} cannot ensure the existence of bubbles in our model.
\end{quote}
The model of \citet{MiaoWang2018} appears to have been understood as a rational bubble model. According to our systematic literature search detailed in Appendix \ref{sec:search}, there are 74 papers that are mainly about the theory of bubbles (broadly defined) and cite \citet{MiaoWang2018} in the context of bubbles. Among these 74 papers, 68 cite it in the context of \emph{rational bubbles} specifically. While most of these papers cite \citet{MiaoWang2018} in passing, some explicitly state that it deals with rational bubbles attached to real assets. For instance, \citet*[Footnote 4, p.~S69]{DongMiaoWang2020} (whose model is a rational bubble model) state
\begin{quote}
    Introducing dividends or rents will complicate our analysis without changing our key insights. See \citet{MiaoWang2018} and \citet*{MiaoWangXu2015} for models of rational bubbles attached to assets with dividends or rents.
\end{quote}

In the present comment, contrary to \citet{MiaoWang2018}'s claim, we prove in Proposition \ref{prop:nonexistence} that rational bubbles as speculation do not exist in their model. For this purpose, in Lemma \ref{lem:bubble_cont} we extend the Bubble Characterization Lemma of \citet{Montrucchio2004} to a continuous-time setting, which could be viewed as our technical contribution. Stock prices in \citet{MiaoWang2018} reflect fundamentals and therefore do not contain the aspect of speculation. Therefore, it is incorrect to understand the model of \citet{MiaoWang2018} as dealing with rational bubbles as speculation pioneered by \citet{Samuelson1958}. The development of the macro-finance theory with speculation requires a correct understanding of it.\footnote{As noted earlier, as \citet{Samuelson1958} discovered, the transversality condition for asset pricing has an important economic meaning, \ie, speculation. \citet{MiaoWang2018} do not mathematically verify the violation of it in the first place.}

To make these points clear, this comment proceeds as follows. In \S\ref{subsec:definition}, following the textbook treatment of \citet[\S13.6]{Miao2014Dynamics}, we review the precise and standard definition of ``rational bubbles'' in a self-contained way so that the general audience can follow. A rational bubble is a situation in which the transversality condition for asset pricing is violated, and hence the asset price exceeds its fundamental value. In brief, it is a bubble as speculation backed by nothing. In \S\ref{subsec:characterization}, we discuss the simple yet powerful Bubble Characterization Lemma of \citet{Montrucchio2004}, which provides a necessary and sufficient condition for the existence of rational bubbles in economies without aggregate uncertainty. In \S\ref{subsec:difficulty}, we explain the fundamental difficulty in attaching rational bubbles to dividend-paying assets following the results of \citet{Kocherlakota1992} and \citet{SantosWoodford1997}. In \S\ref{sec:MW}, which is the main part of this comment, we prove in Proposition \ref{prop:nonexistence} the nonexistence of rational bubbles in the \citet{MiaoWang2018} model.

\section{Rational bubbles as speculation}\label{sec:bubble}

\subsection{Formal definitions}\label{subsec:definition}

The formal definition of rational bubbles was given by \citet{SantosWoodford1997}. Here we follow the textbook treatment of \citet[\S13.6]{Miao2014Dynamics} nearly verbatim. Consider an infinite-horizon economy with a homogeneous good and time indexed by $t=0,1,\dotsc$.\footnote{The continuous-time case is briefly discussed in Appendix \ref{sec:cont}.} Let $\pi_t$ denote a state price deflator.\footnote{If markets are incomplete, which is considered in \citet{SantosWoodford1997}, $\pi_t$ need not be unique. The subsequent discussion holds regardless of market completeness.} For instance, in a deterministic economy, $\pi_t$ is the date-0 price of a zero-coupon bond with maturity $t$. Consider an asset with infinite maturity that pays dividend $D_t\ge 0$ and trades at ex-dividend price $P_t$, both in units of the time-$t$ good. Then the no-arbitrage asset pricing equation\footnote{In specific models, this no-arbitrage condition must be derived from individual optimization problems, which usually comes from the first-order condition. For instance, in \citet{HiranoStiglitz2024}, land can be used as a means of savings and as collateral for borrowing, in which case a land collateral premium arises. They derive the no-arbitrage equation in a consistent manner with individual optimization and show that land prices can be written in the same form as \eqref{eq:P_iter} and the collateral premium is included in the discount rate (see \citet[Appendix O.4]{HiranoStiglitz2024}). They show that higher land prices due to the increased land collateral premium is different from rational land price bubbles as speculation backed by nothing. It is important to note that the discussion following \eqref{eq:noarbitrage_pi} is model-free.} is given by
\begin{equation}
    \pi_t P_t = \E_t[\pi_{t+1}(P_{t+1}+D_{t+1})]. \label{eq:noarbitrage_pi}
\end{equation}
Solving this equation forward by repeated substitution (and applying the law of iterated expectations) yields
\begin{equation}
    \pi_tP_t=\E_t\sum_{s=t+1}^T \pi_sD_s+\E_t[\pi_TP_T]. \label{eq:P_iter}
\end{equation}
Because all terms are nonnegative, the sum in \eqref{eq:P_iter} from $s=t+1$ to $s=T$ is
\begin{enumerate*}
    \item increasing in $T$ and
    \item bounded above by $\pi_tP_t$,
\end{enumerate*}
so it converges almost surely as $T\to\infty$. Therefore the \emph{fundamental value} of the asset
\begin{equation}
    V_t\coloneqq \frac{1}{\pi_t}\E_t\sum_{s=t+1}^\infty \pi_sD_s \label{eq:Vt}
\end{equation}
is well-defined, and letting $T\to\infty$ in \eqref{eq:P_iter}, we obtain $P_t=V_t+B_t$, where we define the \emph{asset price bubble} as
\begin{equation}
    B_t\coloneqq \lim_{T\to\infty} \frac{1}{\pi_t}\E_t[\pi_TP_T]\ge 0.\label{eq:Bt}
\end{equation}
That is, an asset price bubble is equal to the difference between the market price of the asset and its fundamental value (\ie, the present value of dividends). By definition, there is no bubble at time $t$ if and only if
\begin{equation}
    \lim_{T\to\infty} \E_t[\pi_TP_T]=0. \label{eq:TVC}
\end{equation}
This is the mathematical formalization of the idea explained in \citet{Stiglitz1990}. The condition \eqref{eq:TVC} is called the \emph{transversality condition} for asset pricing.\footnote{\citet{Samuelson1958} does not use the word ``transversality condition'' but correctly understands that no-arbitrage alone does not pin down asset prices, writing (p.~470) ``We never seem to get enough equations: lengthening our time period turns out always to add as many new unknowns as it supplies equations''. \citet{SantosWoodford1997} attribute the transversality condition to the unpublished manuscript of \citet{Scheinkman1977}, which we were unable to find. Appendix \ref{subsec:TVC} further clarifies the transversality condition.}

Three remarks are in order. 
\begin{enumerate*}
    \item First, the economic meaning of the bubble component $B_t$ in \eqref{eq:Bt} is that it captures a speculative aspect, that is, agents buy the asset now for the purpose of resale in the future, rather than for the purpose of receiving dividends. When the transversality condition for asset pricing \eqref{eq:TVC} holds, the aspect of speculation becomes negligible and asset prices are determined only by factors that are backed in equilibrium, namely future dividends. On the other hand, if $\lim_{T\to\infty}\E_t[\pi_TP_T]>0$, equilibrium asset prices contain a speculative aspect.
    \item Second, if $D_t=0$ for all $t$ (\citet{Samuelson1958}'s case of \emph{money} or \emph{pure bubble}), the fundamental value of the asset is zero and there is only an aspect of speculation. In this case, $P_t>0$ if and only if $\lim_{T\to\infty}\E_t[\pi_TP_T]>0$, \ie, the current price of the asset will be positive in equilibrium if and only if one expects that one will be able to resell the asset at a positive price in the future. The definition of a rational bubble is the same for money and for real assets yielding positive dividends: it is the speculative component of the asset price.\footnote{It is often claimed that asset bubbles are associated with human irrationality. The rational bubble literature does not deny that, but rather stresses that even if agents hold rational expectations and common knowledge about assets, bubbles can still arise as an equilibrium outcome. Take \citet{Samuelson1958}'s famous model of money. Unless everyone in the economy infinitely far into the future accepts money and expects so, money cannot circulate at present. This expectation formation requires rationality. In other words, the circulation of money is an outcome of the pursuit of rationality. The \citet{Samuelson1958} model clearly illustrates an important insight of rational bubbles, \ie, rationality and speculation.} 
    \item Third, although the definition is the same, there is a discontinuity in proving the existence of a bubble between the cases with $D_t=0$ and $D_t>0$. In other words, as we explain in \S\ref{subsec:difficulty}, it is well known from \citet{Kocherlakota1992} and \citet{SantosWoodford1997} that there is a fundamental difficulty in generating a bubble attached to an asset with $D_t>0$.
\end{enumerate*}

\subsection{Bubble Characterization Lemma}\label{subsec:characterization}

To prove the existence of rational bubbles, we need to prove $P>V$, or equivalently, verify the violation of the transversality condition for asset pricing \eqref{eq:TVC}. For an asset that pays no dividends ($D=0$, pure bubble), because the fundamental value is necessarily zero, showing $P>0$ suffices. However, for dividend-paying assets, \ie, real assets such as stocks, land, and housing, the verification of the violation of the transversality condition is not easy because it is cumbersome to calculate the state price deflator $\pi_t$. Fortunately, in economies without aggregate uncertainty, there is a very simple characterization due to \citet{Montrucchio2004}.\footnote{\citet{Montrucchio2004} also considers the case with aggregate uncertainty, but the characterization is less sharp. Without aggregate uncertainty, the condition is ``necessary and sufficient'' as in Lemma \ref{lem:bubble}, whereas with aggregate uncertainty, there are separate ``necessary'' and ``sufficient'' conditions.}

\begin{lem}[Bubble Characterization, \citealp{Montrucchio2004}, Proposition 7]\label{lem:bubble}
In an economy without aggregate uncertainty, if $P_t>0$ for all $t$, the asset price exhibits a rational bubble if and only if $\sum_{t=1}^\infty D_t/P_t<\infty$.
\end{lem}

See \citet{HiranoTodaNecessity} for a simple proof of the Bubble Characterization Lemma as well as its many applications.

\subsection{Fundamental difficulty in attaching rational bubbles to dividend-paying assets}\label{subsec:difficulty}

A primordial (though overlooked) implication of the Bubble Characterization Lemma is that if the price-dividend ratio $P_t/D_t$ (or the dividend yield $D_t/P_t$) converges to a positive constant, then necessarily $\sum_{t=1}^\infty D_t/P_t=\infty$, so there is no rational bubble by Lemma \ref{lem:bubble}. This implies that in a model with dividend-paying assets, irrespective of the model setup, a rational bubble as speculation can never occur if the price-dividend ratio (or the dividend yield) converges to a positive constant. In particular, in models where $(P_t,D_t)=(P,D)$ is constant, a rational bubble can arise only if $D=0$, \ie, money or pure bubble. We can easily see this result even without appealing to Lemma \ref{lem:bubble}. Consider an asset with price $P>0$ and dividend $D>0$. Then the gross risk-free rate is $R=(P+D)/P>1$, so the fundamental value of the asset is
\begin{equation*}
    V=\sum_{s=1}^\infty R^{-s}D=\frac{D/R}{1-1/R}=\frac{D}{R-1}=\frac{D}{\frac{P+D}{P}-1}=P.
\end{equation*}
Since $P=V$, there is no rational bubble.

More generally, rational bubbles cannot arise if dividends consist of a non-negligible fraction of aggregate endowments. This fundamental difficulty of attaching a rational bubble to dividend-paying assets follows from the results of \citet{Kocherlakota1992} and \citet{SantosWoodford1997}, which are often overlooked despite the fact that they are critically important for the proof of the existence of rational bubbles. \citet[Proposition 4]{Kocherlakota1992} shows that in infinite-horizon models with rational bubbles, there must exist an agent whose asset holdings fluctuate between two positive numbers infinitely often and that the present value of the endowments of such an agent is infinite.\footnote{The proof contained an error and was corrected by \citet{KocherlakotaToda2023JET}.} \citet[Theorem 3.3]{SantosWoodford1997} significantly extend this result under incomplete markets and borrowing constraints and show that if the present value of the aggregate endowment is finite, then the price of an asset in positive net supply or with finite maturity equals its fundamental value. Furthermore, their Corollary 3.4 (together with Lemma 2.4) shows that, when the asset pays nonnegligible dividends relative to the aggregate endowment, rational bubbles as speculation are impossible. See \citet[\S3.4]{HiranoToda2024JME} for an accessible discussion of these results.

The bubble impossibility results of \citet{Kocherlakota1992} and \citet{SantosWoodford1997} have been understood as a fundamental difficulty in attaching bubbles to dividend-paying assets. Perhaps due to this fundamental difficulty, the rational bubble literature (see \citet{MartinVentura2018} and \citet{HiranoToda2024JME} for reviews) has almost exclusively studied pure bubble models with $D_t=0$ and considered attaching bubbles to dividend-paying assets in a natural way as a long-standing and important open question. Therefore, if one claims to prove the existence of a rational bubble attached to an asset with $D_t>0$, one needs to verify so in a consistent manner with \citet{Kocherlakota1992}, \citet{SantosWoodford1997}, and \citet{Montrucchio2004}. Otherwise, it does not mean one has proved it.

\section{Nonexistence of rational bubbles in \texorpdfstring{\citet{MiaoWang2018}}{}}\label{sec:MW}

\citet{MiaoWang2018} (henceforth MW) claim to ``provide a theory of rational stock price bubbles'' (quoted from their abstract) in an equilibrium model in relation to financial conditions. Here is a brief model description. There are a continuum of firms maximizing shareholder value, and investment opportunities arrive stochastically according to independent Poisson processes. Upon the arrival of an investment opportunity, a firm may invest subject to a collateral constraint that involves the firm value. Otherwise, the model is a standard neoclassical growth model (representative households, neoclassical production function, no capital adjustment costs, etc.).\footnote{Because \citet{MiaoWang2018} provide a micro-foundation for the collateral constraint and derive the continuous-time model by carefully taking the limit of a discrete-time model, the exposition is rather involved. See \citet{Sorger2020} for a concise model description.} MW show that the firm value (stock price) equals $V(K)=QK+B$, where $K$ is capital and $Q>0$ and $B\ge 0$ are constants. Note that because investment opportunities arrive to firms stochastically, a firm with no capital at present could have a positive value $V(0)=B$ in anticipation of the arrival of future investment opportunities. Importantly, MW (p.~2601) ``interpret'' $QK$ as the fundamental value and $B$ as the bubble. Based on the way MW write their paper, it is clear that they view their model as a rational bubble model. However, using the Bubble Characterization Lemma in \S\ref{subsec:characterization}, it is straightforward to show that in the MW model, the stock price equals the fundamental value and hence there is no speculative aspect.

\begin{prop}\label{prop:nonexistence}
Rational bubbles in the sense of \S\ref{sec:bubble} do not exist in the model of \citet{MiaoWang2018}.
\end{prop}

\begin{proof}
In the MW model, there is a continuum of firms of unit measure subject to only idiosyncratic risks. Let $P_{it},V_{it}$ be the stock price and fundamental value of firm $i$ at time $t$. Similarly, let $P_t,V_t$ be the price and fundamental value of the aggregate stock market. By the discussion in \S\ref{subsec:definition}, we have $P_{it}-V_{it}\ge 0$. Aggregating across firms, we have $P_t-V_t\ge 0$, with equality if and only if $P_{it}=V_{it}$ for $i$ almost surely. Therefore to prove the nonexistence of rational stock price bubbles in individual firms, it suffices to prove the nonexistence of rational bubbles in the aggregate stock market.

Because there is no aggregate uncertainty, aggregate dividend is paid out continuously according to $\diff F_t=D_t\diff t$ in the notation of Lemma \ref{lem:bubble_cont} in Appendix \ref{sec:cont}. Furthermore, the aggregate stock price $P_t$ and the dividend $D_t$ converge to constants $P,D>0$. (The positivity of dividend is discussed on p.~2608.) In particular, we have $\int_0^\infty D_t/P_t\diff t=\infty$. By Lemma \ref{lem:bubble_cont}, there exist no rational bubbles in either the aggregate stock market or individual firms.
\end{proof}

The problem with \citet{MiaoWang2018} is that the flow of logic is the opposite to rational bubble models, yet they write their paper as if they have resolved the fundamental difficulty in rational bubble models described in \S\ref{subsec:difficulty}. To clarify, in standard rational bubble models as in \S\ref{subsec:definition}, the flow of logic is that
\begin{enumerate*}
    \item the asset pricing equation \eqref{eq:noarbitrage_pi} is derived from individual optimization, and then
    \item whether there is a rational bubble or not (whether $P>V$ or $P=V$) is verified based on the transversality condition \eqref{eq:TVC}.
\end{enumerate*}
In contrast, in MW, the flow of logic is that
\begin{enumerate*}
    \item[(i')] the bubble component $B_t$ is defined based on an ``interpretation'' of their own, and then
    \item[(ii')] the differential equation that $B_t$ satisfies (which is Equation (20) in MW) is interpreted as an asset pricing equation.
\end{enumerate*}
What concerns us is that MW repeatedly use the word ``rational'' (20 times in total) and compare their results to the rational bubble literature as quoted in the introduction, without acknowledging the fact that the definition of rational bubbles is completely different. In other words, MW claim to have solved an important problem in the literature by changing the problem itself.  

Several papers study models similar to MW. In these papers, for justifying bubbles, an expression like ``we interpret $B$ as a bubble'' appears, without reference to the formal definition and the economic meaning of rational bubbles of the sort in \S\ref{sec:bubble}. Table \ref{t:expression} lists these expressions in MW (whose working paper circulated since 2011) and the subsequent follow-up papers that we found based on our systematic literature search detailed in Appendix \ref{sec:search}.

\begin{table}[!htb]
    \centering
    \begin{tabular}{llr}
    \toprule
    Paper & Expression & Page \\
    \midrule
    \citet{MiaoWang2012} & ``we require $b_{it}>0$ and interpret it as a bubble'' & 85 \\
    \citet{MiaoWang2014} & ``We interpret $B_{it}>0$ as the bubble component'' & 157 \\
    \citet{MiaoWang2015} & ``We may interpret $B_t$ as a bubble component'' & 772 \\
    \citet{MiaoWangXu2015} & ``we may interpret this component as a bubble'' & 608 \\
    \citet{MiaoWangXu2016} & ``$B_t>0$ for all $t$ and interpret it as a bubble'' & 283 \\
    \citet{MiaoWang2018} & ``we interpret $B_t$ as a bubble component'' & 2601 \\
    \citet{ChevallierElJoueidi2019} & ``$b_t\neq 0$ is the bubble term'' & 122 \\
    \citet{Sorger2020} & ``$q(t)\neq 0$ is said to contain a bubble'' & 525 \\
    \citet{He2021} & ``the second term is viewed as the
bubble'' & 251 \\
    \citet{Ikeda2022} & ``interpreted as a liquidity bubble'' & 1579 \\
    \bottomrule
    \end{tabular}
    \caption{Typical expressions in models of $V(K)=QK+B$.}
    \label{t:expression}
\end{table}

These papers appear to have been understood as models of rational asset price bubbles attached to real assets (Appendix \ref{sec:search}). However, as discussed in \S\ref{subsec:difficulty}, if one claims to prove the existence of rational bubbles attached to real assets with $D_t>0$, one needs to prove it in a consistent manner with \citet{Kocherlakota1992}, \citet{SantosWoodford1997}, and \citet{Montrucchio2004}.\footnote{Even in the papers listed in Table \ref{t:expression}, the violation of the transversality condition for asset pricing \eqref{eq:TVC} is not verified. Hence the existence of rational bubbles is not proved (indeed, rational bubbles do not exist by the same argument as Proposition \ref{prop:nonexistence}), despite the fact that the rational bubble literature is frequently cited and compared.} Moreover, bubbles in dividend-paying assets can never occur so long as the price-dividend ratio (or the dividend yield) converges to a positive constant. (See the discussion in \S\ref{subsec:difficulty} of implications of Lemma \ref{lem:bubble}.) In models in Table \ref{t:expression}, asset prices reflect fundamentals, regardless of the interpretation of the term $B$. In other words, in all papers in Table \ref{t:expression}, no rational bubbles exist. Thus, the way \citet{MiaoWang2018} write their paper is misleading as they claim the existence of a bubble as if it were a rational bubble.

It is more appropriate to interpret \citet{MiaoWang2018} and others as multiple equilibria in asset pricing models, where there are two steady states, one with high stock prices and the other with low stock prices. In both steady states, stock prices always reflect fundamentals, but self-fulfilling expectations determine which steady state is reached.\footnote{\citet{Azariadis1981}, \citet{CassShell1983}, and \citet{Farmer1999} have produced a large literature that studies the effects of self-fulfilling expectations on macroeconomic outcomes including asset prices, which is clearly important. \citet{MiaoWang2018} and others would belong to this literature.} In fact, \citet[p.~772]{MiaoWang2015} state ``one may also interpret it as a self-fulfilling component or a belief component if one wants to avoid using the term ``bubble''''.

\section{Conclusion}

Due to the scattering of diverse views and definitions, the term ``bubble'' tends to be used in different ways across papers and sometimes within the same paper without acknowledging the difference. This is akin to switching the rule between baseball and cricket at will during a game. We think that this dual standard leads to further confusion, distorts the literature, and prevents the healthy development of science. 

Rational bubbles pioneered by \citet{Samuelson1958} mean speculation backed by nothing (see the precise mathematical definition of it in \S\ref{sec:bubble}). \citet{MiaoWang2018} claim to ``provide a theory of rational stock price bubbles'' and their paper appears to have been understood as a model of rational bubbles. Indeed, the way \citet{MiaoWang2018} write their paper gives the impression that they have provided a theory of rational bubbles attached to dividend-paying assets. However, contrary to their claim, as we prove in this comment, no rational bubbles as speculation exist in their model. We would like to stress that the nonexistence result of rational bubbles in \citet{MiaoWang2018} is not just a matter of semantics on the term ``bubble''. Rather, it should be understood that the current state of widespread misunderstanding can mislead the general reader's understanding of rational bubbles as speculation, and worse, it may end up distorting the literature. Both should be avoided.

Bubbles are an important economic phenomenon. Although our primary research interest is the theory of rational bubbles as outlined in \S\ref{sec:bubble}, we have an open mind, and we welcome various approaches such as based on rationality, heterogeneous beliefs, asymmetric information, self-fulfilling expectations, econometrics, empirical considerations, or others. Researchers are free to define a bubble in whatever way they like. In our opinion, a significant merit of the rational bubble approach pioneered by \citet{Samuelson1958}, \citet{Bewley1980}, \citet{Tirole1985}, \citet{ScheinkmanWeiss1986}, \citet{Kocherlakota1992}, \citet{SantosWoodford1997} and outlined in \S\ref{sec:bubble} is that the definition is precise, model-free, and has an important economic meaning, \ie, speculation backed by nothing, as opposed to the more or less model-dependent and ad hoc nature of other approaches. Whatever definition we adopt, we need to be consistent. It is not right to claim that a particular model generates rational bubbles attached to dividend-paying assets using one definition, when in fact there exist no rational bubbles in the same model under a different (and more standard) definition, without acknowledging the differences in the definitions.

The purpose of this comment is to reach a better and mutual understanding about bubbles, not just among experts but also among the general audience. We hope that our comment clarifies the confusion surrounding the theory of rational bubbles as speculation and that the science of bubbles prosper, without ever imploding.

\appendix

\section{Systematic literature search}\label{sec:search}

We use Google Scholar to identify papers that cite \citet{MiaoWang2018} (henceforth MW) or its working paper version titled ``Bubbles and Credit Constraints''. As of the time of search (July 2024), MW was cited 198 times and the working paper version 125 times. Of these citations, we only retain published journal articles written in English and remove duplicate citations, resulting in 125 citations. We then checked each paper to see how MW is cited. Of these 125 papers, 74 are mainly about the theory of bubbles (broadly defined, not necessarily rational bubbles), 96 cite MW in the context of bubbles (as opposed to other contexts such as credit constraints, investment, multiple equilibria, etc.), and 71 cite MW in the context of \emph{rational bubbles} specifically. (We regard the citation as rational bubbles if the paper explicitly mentions ``rational bubble'' or cites MW along with other papers on rational bubbles.) All 74 papers that are mainly about the theory of bubbles (broadly defined) cite MW in the context of bubbles. Among these 74 papers, 68 cite MW in the context of rational bubbles specifically. For more details, see the spreadsheet available at \url{https://alexisakira.github.io/files/MW.xlsx}.

\section{Bubble characterization in continuous-time}\label{sec:cont}

This appendix extends Lemma \ref{lem:bubble} to continuous-time.

First, consider the discrete-time setting in \S\ref{subsec:definition} but without aggregate uncertainty. If the asset is risk-free, taking the unconditional expectations of \eqref{eq:noarbitrage_pi} and setting $q_t=\E[\pi_t]>0$ (which equals the date-0 price of a zero-coupon bond with maturity $t$), we obtain
\begin{equation}
    q_tP_t = q_{t+1}(P_{t+1}+D_{t+1}). \label{eq:noarbitrage}
\end{equation}
Then by the same argument as in \S\ref{subsec:definition} and using $q_0=1$, we obtain
\begin{equation}
    P_0=\sum_{t=1}^T q_tD_t+q_TP_T, \label{eq:P0_iter}
\end{equation}
and there is no bubble if the transversality condition $\lim_{T\to\infty}q_TP_T=0$ holds. (See \citet[\S2]{HiranoToda2024JME} and \citet[\S2]{HiranoTodaNecessity} for details.)

Now consider the continuous-time model without aggregate uncertainty. Let $F_t\ge 0$ be the cumulative dividend payout of an asset up to time $t$ (so $t\mapsto F_t$ is weakly increasing and right-continuous), and $P_t$ be its ex-dividend price. In a small time interval $(t,t+\Delta]$, the dividend is $F_{t+\Delta}-F_t$. Therefore the continuous-time counterpart of the no-arbitrage condition \eqref{eq:noarbitrage} is
\begin{equation*}
    q_tP_t=q_{t+\Delta}(P_{t+\Delta}+F_{t+\Delta}-F_t).
\end{equation*}
Subtracting $q_{t+\Delta}P_{t+\Delta}$ from both sides and taking the limit $\Delta\to 0$, we obtain
\begin{equation}
    -\diff(q_tP_t)=q_t\diff F_t.\label{eq:noarbitrage_cont}
\end{equation}
Integrating both sides from $t=0$ to $t=T$ and using $q_0=1$, we obtain
\begin{equation}
    P_0=\int_0^T q_t\diff F_t+q_TP_T, \label{eq:P_iter_cont}
\end{equation}
which is the continuous-time counterpart of \eqref{eq:P0_iter}. Multiplying both sides of \eqref{eq:TVC} by $\pi_t$, taking the unconditional expectation, and letting $T\to\infty$, we see that the transversality condition for asset pricing is $\lim_{T\to\infty}q_TP_T=0$. Under this condition, by \eqref{eq:P_iter_cont}, the asset price equals its fundamental value $V_0\coloneqq \int_0^\infty q_t\diff F_t$.

The following lemma provides a continuous-time counterpart of Lemma \ref{lem:bubble}.

\begin{lem}[Bubble characterization in continuous-time]\label{lem:bubble_cont}
In an economy without aggregate uncertainty, if $P_t>0$ for all $t$, the asset price exhibits a rational bubble if and only if $\int_0^\infty \frac{\diff F_t}{P_t}<\infty$.
\end{lem}

\begin{proof}
Dividing both sides of \eqref{eq:noarbitrage_cont} by $q_tP_t>0$ and integrating from $t=0$ to $t=T$, we obtain
\begin{equation*}
    \int_0^T \frac{\diff F_t}{P_t}=-\int_0^T \frac{\diff (q_tP_t)}{q_tP_t}=-\int_0^T\diff \log (q_tP_t)=\log q_0P_0-\log (q_TP_T).
\end{equation*}
Rearranging terms, we obtain
\begin{equation*}
    q_TP_T=q_0P_0\exp\left(-\int_0^T \frac{\diff F_t}{P_t}\right).
\end{equation*}
Letting $T\to\infty$, the transversality condition $\lim_{T\to\infty}q_TP_T=0$ holds if and only if $\int_0^\infty \frac{\diff F_t}{P_t}=\infty$.
\end{proof}

\section{Transversality condition(s)}\label{subsec:TVC}

The condition \eqref{eq:TVC} is called the \emph{transversality condition} for asset pricing. This is another unfortunate use of language for two meanings. In the asset pricing literature, the transversality condition \eqref{eq:TVC} means that the present value of the asset in the far distant future approaches zero and hence there is no speculative motive. In infinite-horizon optimal control theory, the transversality condition is a necessary or sufficient optimality condition that states that the marginal benefit of deviating from the current plan at infinity is zero. See \citet{EkelandScheinkman1986}, \citet{Kamihigashi2002}, or the forthcoming textbook of \citet[\S15.3]{TodaEME} for more discussion. Interestingly, in infinite-horizon models, in the absence of financial frictions, the transversality condition \emph{for optimality} rules out rational bubbles (so the transversality condition \emph{for asset pricing} is satisfied): see \citet{Kamihigashi1998} for this result and \citet[\S3.4]{HiranoToda2024JME} for a related discussion. However, once financial frictions (such as shortsales constraints) are introduced, it is well known that rational bubbles are possible, so the transversality condition \emph{for optimality} is satisfied yet the transversality condition \emph{for asset pricing} is violated \citep{Bewley1980,Kocherlakota1992}. In contrast, in two-period overlapping generations models such as \citet{Samuelson1958} and \citet{Tirole1985}, because individual agents have a finite horizon, only the transversality condition for asset pricing is relevant.

This language use could be very confusing for researchers that are new to the literature. One should clearly distinguish the two meanings but it should be clear from the context. If not, it is safer to explicitly mention the transversality condition ``for optimality'' or ``for asset pricing''. In what follows, the transversality condition is always ``for asset pricing''.

\printbibliography
	
\end{document}